%% file: main.tex
\documentclass[runningheads]{llncs}
\usepackage[T1]{fontenc}
\usepackage{graphicx}

\usepackage{quiver}
\usepackage{enumitem}
\usepackage{booktabs}
\usepackage{paralist}
\usepackage{multirow}
\usepackage{scalerel}
\newcommand\mathbox[1]{\mathord{\ThisStyle{%
  \fboxsep3\LMpt\relax\kern1\LMpt\fbox{$\SavedStyle#1$}\kern1\LMpt}}}
\usepackage{url}

\usepackage{enumerate}
\usepackage{nicematrix}
\usepackage{tikz,pgfplots}
\usetikzlibrary{fit}
\usepackage{pgfplotstable}
\usetikzlibrary{patterns}
\usepackage{wrapfig}
\usepackage{array} % For better column definitions
\pgfplotsset{compat=1.16}

\makeatletter
\newcommand\resetstackedplots{
\makeatletter
\pgfplots@stacked@isfirstplottrue
\makeatother
\addplot [forget plot,draw=none] coordinates{(192, 0) (768, 0) (3072, 0) (12288, 0) (49152, 0) (196608, 0)};
}
\makeatother

\definecolor{pgfplotsdarkblue}{RGB}{0,0,102}
\definecolor{pgfplotsmidblue}{RGB}{0,0,255}
\definecolor{pgfplotslightblue}{RGB}{152,152,255}

\begin{document}

\def\thefootnote{*}\footnotetext{These authors contributed equally to this work.}
\title{Private, Efficient and Scalable Kernel Learning for Medical Image Analysis}
\author{Anika Hannemann\inst{1,2,\thefootnote{*}} \and
Arjhun Swaminathan\inst{3,4,\thefootnote{*}} \and
Ali Burak Ünal\inst{3,4} \and
Mete Akgün\inst{3,4}}

\authorrunning{A. Hannemann and A. Swaminathan et al.}

\institute{
    $^1$ Department of Computer Science, Leipzig University, Leipzig, Germany \\
    $^2$ Center for Scalable Data Analytics and Artificial Intelligence (ScaDS.AI) Dresden/Leipzig, Germany \\
    $^3$ Medical Data Privacy and Privacy-preserving Machine Learning (MDPPML), University of Tübingen, Tübingen, Germany \\
    $^4$ Institute for Bioinformatics and Medical Informatics (IBMI), University of Tübingen, Tübingen, Germany \\
    \email{anika.hannemann@cs.uni-leipzig.de, arjhun.swaminathan@uni-tuebingen.de}
}

\maketitle              % typeset the header of the contribution
\begin{abstract}
Medical imaging is key in modern medicine. From magnetic resonance imaging (MRI) to microscopic imaging for blood cell detection, diagnostic medical imaging reveals vital insights into patient health. To predict diseases or provide individualized therapies, machine learning techniques like kernel methods have been widely used. Nevertheless, there are multiple challenges for implementing kernel methods. Medical image data often originates from various hospitals and cannot be combined due to privacy concerns, and the high dimensionality of image data presents another significant obstacle. While randomised encoding offers a promising direction, existing methods often struggle with a trade-off between accuracy and efficiency.
Addressing the need for efficient privacy-preserving methods on distributed image data, we introduce OKRA (Orthonormal K-fRAmes), a novel randomized encoding-based approach for kernel-based machine learning. This technique, tailored for widely used kernel functions, significantly enhances scalability and speed compared to current state-of-the-art solutions. Through experiments conducted on various clinical image datasets, we evaluated model quality, computational performance, and resource overhead. Additionally, our method outperforms comparable approaches.
\end{abstract}

\input{1-introduction}
\input{2-related_work}
\input{3-methodology}

\input{4-privacy}
\input{5-experiments}
\input{6-conclusion}

\bigskip \noindent {\small{\textbf{Acknowledgments.} This study is supported by the German Federal Ministry of Research and Education (BMBF), under project number 01ZZ2010.}}


\bibliographystyle{splncs04}
\bibliography{9-bibliography}

\end{document}

%% file: 1-introduction.tex
\section{Introduction}

Medical imaging is crucial in modern medicine, offering insights into the human body and revolutionizing healthcare by aiding in disease detection, treatment planning, and patient monitoring. Kernel-based machine learning, or kernel learning, has become widely accepted in medical image analysis for identifying complex patterns by mapping data to higher-dimensional spaces \cite{seo2020machine,hasan2021review}. Techniques like kernel-based Support Vector Machines (SVM), Principal Component Analysis (PCA), Canonical Correlation Analysis (CCA), Gaussian processes, and k-means are particularly effective for smaller \cite{dou2023machine}, high-dimensional datasets and offer greater interpretability \cite{seo2020machine}. 

However, medical data is often distributed across various institutions, complicating the implementation of kernel learning due to privacy risks. Health data can re-identify individuals \cite{el2011systematic}, and regulations like GDPR prohibit sharing sensitive data. To tackle this, privacy preserving techniques could be employed. However, traditional privacy techniques, such as Homomorphic Encryption (HE) \cite{gentry2009fully} and Secure Multiparty Computation (SMPC) \cite{yao1982protocols}, face operational challenges, while Differential Privacy (DP) introduces model inaccuracies \cite{bagdasaryan2019differential,hannemann2024differentially}. 

Moreover, many healthcare institutions rely on cloud services for computational tasks \cite{al2020health}. With this in mind, our work proposes kernel-based learning on distributed medical data within a federated architecture using a semi-honest central server. This approach aligns with real-world scenarios where institutions aim to perform joint analysis for higher accuracy while maintaining data privacy and lacking on-premise computational resources. Specifically, we employ a one-shot federated learning approach, where a global model is computed in a single communication round \cite{guha2019one,griebel2015scoping}. 

This method ensures the central server cannot derive sensitive information from the data.
We introduce OKRA (Orthonormal K-fRAmes), which uses randomized encoding to project input data into higher-dimensional spaces. This facilitates the training of kernel learning models as if the data were centralized, with lower computational and communication costs compared to traditional privacy techniques. Unlike the comparable state-of-the-art solutions, OKRA requires less computation time, especially for high-dimensional data.

We make three contributions: 
\begin{compactenum}
\item We present the OKRA methodology to privately compute kernel functions on distributed data. 
\item We report a privacy analysis confirming the ability of OKRA to maintain the privacy of both input data and the size of medical images.
\item We evaluate OKRA through extensive experimentation involving various combinations of datasets, kernels and Machine Learning techniques. 
\end{compactenum}

\textbf{Paper structure}: Section \ref{sec:related_work} introduces related work.
Section~\ref{sec:methods} describes our methodology, followed by a privacy analysis in Section~\ref{sec:privacy}.
Section~\ref{sec:experiments} presents our experimental results, followed by conclusion in Section~\ref{sec:conclusion}.

%% file: 2-related_work.tex
\section{Related Work}
\label{sec:related_work}

\subsection{One-shot Federated Learning}
Federated Learning (FL) \cite{mcmahan2017communication} allows distributed nodes to train models on local data while preserving privacy. In traditional FL, nodes iteratively share local model parameters with a central server, which aggregates them to form a global model. However, this iterative communication introduces latency and security risks, such as data/model poisoning and membership inference attacks \cite{mothukuri2021survey}.

In contrast, one-shot federated learning approaches such as \cite{guha2019one}, restrict communication to a single round. Each participant sends its relevant parameters to the central server once, which subsequently undertakes the role of model training. For our methodology, the central server computes relevant kernels with a single round of communication.

\subsection{Privacy Preserving Techniques}
FL incorporates various techniques to safeguard the privacy of input data:
\paragraph{Cryptography-based Approaches:}
These primarily constitute the usage of HE and SMPC. HE aims to protect privacy of aggregated models by encrypting the parameters of the local models. However, this is computationally intensive, often leading to slower processing speeds, which can be a bottleneck in real-time applications \cite{costache2016ring}. Meanwhile, SMPC splits the data into secret shares distributed among participants for shared computation. Despite being less computationally intensive compared to HE, it requires continuous communication, often resulting in high execution times \cite{mugunthan2019smpai,chen2020privacy}. 
\paragraph{Differential Privacy:} FL methods using DP introduce noise to individual models to enhance privacy, making it challenging to deduce the original model or identify a specific data point's membership. However, the noise added is irremovable and often impacts model accuracy \cite{adnan2022federated,bagdasaryan2019differential,hannemann2024differentially}.
\paragraph{Randomized Encoding-based Approaches:}
Contrary to more computationally demanding approaches, these methods present an efficient way of preserving data privacy, such as maintaining relationships between dot products \cite{chen2011geometric,unal2019framework}. Randomized encoding approaches \cite{applebaum2006computationally,chen2007towards,lin2015secure,swaminathan2024pp} use a variety of techniques, including but not limited to geometric perturbations. Geometric perturbations introduce structured noise to the data. This noise can be factored out, thus differing from differential privacy methods in the context of preserving model accuracy.  It should be noted that the security of randomized encoding-based approaches is dependent on the specific problem at hand since they are tailored to meet the unique demands and constraints of a privacy-preserving task.

\subsection{Privacy-preserving Kernel Learning using Randomized Encoding}
The integration of randomized encoding methods into kernel learning presents a promising avenue, particularly in their cost-effectiveness compared to traditional approaches. This subsection delves into two noteworthy contributions in this domain that achieve exact model predictions: ESCAPED \cite{unal2021escaped} and FLAKE \cite{hannemann2023privacy}, each offering distinct methodologies and implications.

ESCAPED makes use of randomized encoding within a multi-party computational framework for kernel computations. Its core innovation lies in enabling secure, collaborative computing without compromising individual data privacy. However, this approach necessitates intricate communication channels among all participating data providers. This requirement becomes particularly cumbersome with the introduction of new data entities, posing significant logistical challenges in scalable deployments. On the other hand, FLAKE works towards optimizing communication overhead in kernel-based learning. It employs a one-shot federated learning architecture, requiring data providers to only communicate once with a central server. However, this methodology encounters its own set of challenges, particularly with computational efficiency and scalability. It especially demonstrates limitations in processing high-dimensional data, a common characteristic in complex datasets such as medical images. 

In this work, we present an alternative to current randomized encoding methods for kernel computation in machine learning. 
We propose a one-shot algorithm, hence bypassing ESCAPED's communication overheads, while scaling better than FLAKE to accommodate higher dimensional datasets.

%In this work, we present an alternative to current randomized encoding methods for kernel computation in machine learning.  We propose a refined one-shot algorithm, which significantly improves scalability and computational efficiency, particularly for high-dimensional data. This approach not only optimizes processing but also maintains robust privacy preservation, offering a substantial improvement over existing techniques in both scope and performance.

%% file: 3-methodology.tex
\section{Methodology} \label{sec:methods}

In this section, we delve into the methodology of our approach, focusing on private kernel computation in a distributed environment. First, we detail the adversarial architecture we operate in. In our architecture, for simplicity, we consider participants Alice, Bob, and Charlie, along with a central server orchestrating the training of a machine learning model, as described in Figure \ref{fig:overview}. However, this can be extended to any number of participants. We operate under the assumption of semi-honest participants and a non-colluding central server. A semi-honest party follows the protocol correctly but may attempt to infer additional information. This is in line with existing state-of-the-art \cite{unal2021escaped,hannemann2023privacy}.

\begin{figure*}[h]
    \centering
    \includegraphics[width=0.9\linewidth, trim={2 2 120 200}, clip]{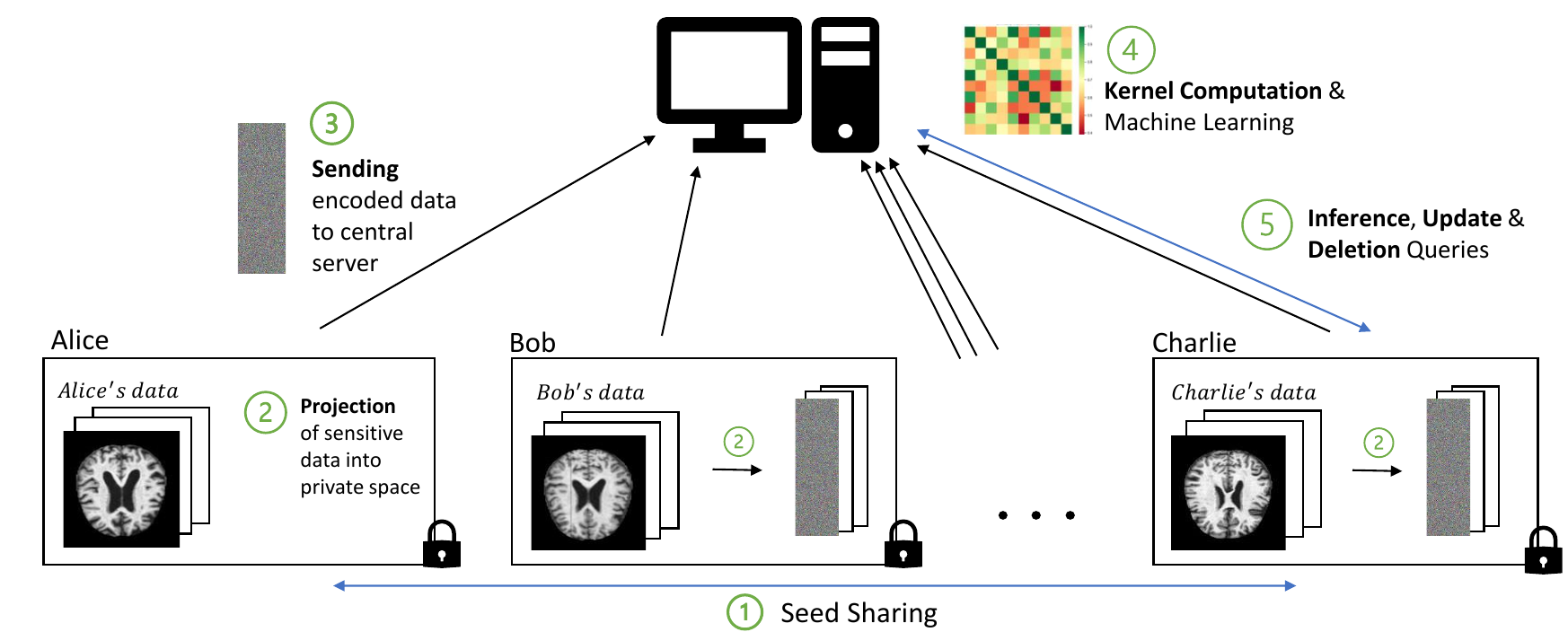}
    %\vskip -0.1in
    \caption{Overview of OKRA.}
    \label{fig:overview}
    \vskip -0.2in
\end{figure*}

Our method must satisfy four key requirements, with \textbf{privacy} being the first: Neither the central server, nor the corrupted participants can learn a non-corrupt participant's data, and the image sizes remain concealed from the server. Second, \textbf{correctness} is crucial, with the model's accuracy needing to align that of a centrally trained model. Third, \textbf{efficiency} is important, as communication costs and execution times should be practical even when input data scales up. Finally, \textbf{data adaptability} is essential, allowing the model to accommodate new data and participants with minimal overhead.

%We now outline the four pivotal requirements that our methodology must satisfy. 

%\textbf{Privacy:} Neither the central server, nor the corrupted participants can learn a non-corrupt participant's data, and the image sizes remain concealed from the server.

%\textbf{Correctness:} The accuracy of the model must align with that of a model trained on centralized data. 

%\textbf{Efficiency:} Communication costs and execution times should be practical for real world scenarios, especially when scaling up the dimensionality of input data.

%\textbf{Data Adaptability:} The model must accommodate new data and participants without significant overhead. 

We will focus on four of the most widely used kernel functions: Linear,  Gaussian, Polynomial, and Rational Quadratic kernels. We begin by briefly defining them.

\subsection{Preliminaries}
Kernel functions project the data into a Hilbert space using a feature map, and compute the similarity between the data points in this space. The trick here lies in the ability to compute this similarity without explicitly mapping the data - a consequence of Mercer's theorem, which provides necessary and sufficient conditions for a function from an input space to be a 'kernel', ensuring the mapped data resides in a Hilbert space. Below are four widely used kernels of interest: Linear Kernel ($K_{lin}$) \cite{cortes1995support}, Gaussian Kernel ($K_{RBF}$) \cite{boser1992training}, Polynomial Kernel ($K_{poly}$) \cite{scholkopf2002learning}, and Rational Quadratic Kernel ($K_{RQ}$) \cite{williams1995gaussian} -
\begin{align*}
K_{lin}(x,y) &= xy^T, & K_{RBF}(x,y) &= \gamma^2 \exp\left(-\frac{{\lVert x-y \rVert^2}}{2l^2}\right), \\
K_{poly}(x,y) &= (1 + xy^T)^d, & K_{RQ}(x,y) &= \gamma^2 \left(1 + \frac{{\lVert x-y \rVert^2}}{2\alpha l^2}\right)^{-\alpha}.
\end{align*}

Here, $\gamma$ is the amplitude, $l$ the length-scale, $d$ the degree of the polynomial, and $\alpha$ the scale mixture. We now define orthonormal k-frames, that are central to our study.

\begin{definition} \textit{Orthonormal k-frame ~\cite{christensen2003introduction}}
\\ A \emph{k-frame} is an ordered set of $k$ linearly independent vectors in a vector space of dimension $n$, with $k \leq n$. An ordered set of $k$ linearly independent orthonormal vectors is called an \emph{Orthonormal k-frame}. 
\end{definition}

\subsection{The OKRA Methodology}
\subsubsection{Data Preprocessing}
Prior to delving into our methodology, we now detail the preprocessing steps for medical imaging datasets. Each dataset comprises $n$ medical images, formatted into a four dimensional array of size $n \times h \times w \times c,$ where $h$ is the height, $w$ the width, and $c$ the number of color channels. We flatten each image into a 1D array, so that the dataset has the size $n \times (hwc:=f)$. 
Participants Alice, Bob, and Charlie possess preprocessed data matrices $A$, $B$, and $C$ of size  $n_{\mathcal{A}} \times f$, $n_{\mathcal{B}} \times f$, and $n_{\mathcal{C}} \times f$ respectively. While $n_{\mathcal{A}}$, $n_{\mathcal{B}}$, and $n_{\mathcal{C}}$ may differ due to variable data availability among participants, we require all participants contribute more than one image. The participants share a seed using any robust Public Key Infrastructure (PKI) for secure communication. We assume a trusted third party for fair public key distribution, in line with standard privacy-preserving federated learning practices.

\subsubsection{Seed Generation}
A robust Public Key Infrastructure (PKI) ensures every participant possesses the public signing keys of all others. At the onset of the protocol, a participant is arbitrarily designated as the leader, responsible for the generation of a random seed. A secure seed is shared with participants using public-key encryption and verified with digital signatures. We assume a trusted third party for fair public key distribution, in line with standard privacy-preserving federated learning practices \cite{bonawitz2017practical,zhang2020batchcrypt}.

\subsubsection{Randomized Encoding}
The participants Alice $(\mathcal{A})$, Bob $(\mathcal{B})$ and Charlie $(\mathcal{C})$ with their shared seeds generate their own sets of $j$ orthonormal k-frames $({_{\mathcal{A}}}O_1,\ldots,{_{\mathcal{A}}}O_j)$, $({_{\mathcal{B}}}O_1,\ldots,{_{\mathcal{B}}}O_j)$, and $({_{\mathcal{C}}}O_1,\ldots,{_{\mathcal{C}}}O_j)$ respectively, using the seed such that \[_\mathcal{P} O_i=[{_\mathcal{P} v_i^1}^T \ldots {_\mathcal{P} v_i^{k_i}}^T],\] for $\mathcal{P} \in\{\mathcal{A},\mathcal{B},\mathcal{C}\}$. Further, for any $p \leq j$ and $\mathcal{P},\mathcal{Q} \in \{\mathcal{A},\mathcal{B},\mathcal{C}\}$, it holds that \[\left({_\mathcal{P}}O_p^{\dag}\right) \left({_\mathcal{Q}}O_p\right)=I,\]  where $\dag$ denotes the conjugate transpose. Each $_\mathcal{P} v_i^l$ is an element of a corresponding vector space $V_{i}$ defined over $\mathbb{C}$ such that $k_i \leq dim(V_i)$. Further, $ V_{{1}}\oplus \ldots \oplus V_{{j}} = V, \text{ } dim(V)=f+k,\text{ and } k_1+\ldots +k_j=f$. Then we define $_\mathcal{P} \Gamma$ to be the direct sum 
\[{_\mathcal{P}} \Gamma = {_\mathcal{P}} O_1 \oplus \dots \oplus  {_\mathcal{P}} O_j.\]

To enhance privacy, a random permutation matrix $\sigma \in S_{f}$ is generated using the shared seed by all three participants. Here $S_f$ denotes the symmetric group of degree $f$. They then permute the rows of their matrices ${_{\mathcal{A}}}\Gamma,{_{\mathcal{B}}}\Gamma$, and ${_{\mathcal{C}}}\Gamma$ according to this permutation, and subsequently, encode their data to derive $A'=A({_{\mathcal{A}}}\Gamma \circ \sigma)^{\dag} \in \mathbb{C}^{n_{\mathcal{A}} \times (f+k)}$, $B'=B({_{\mathcal{B}}}\Gamma \circ \sigma)^{\dag} \in \mathbb{C}^{n_{\mathcal{A}} \times (f+k)}$, and $C'=C({_{\mathcal{C}}}\Gamma \circ \sigma)^{\dag} \in \mathbb{C}^{n_{\mathcal{C}} \times (f+k)}$. 

Alice, Bob, and Charlie send $A'$, $B'$, and $C'$ to the central server. Let $A_i$ denote the $i^{th}$ row of matrix $A$ - an image in our setting. We propose the following.

\begin{theorem}\label{correctness}
    Given $A'$, $B'$ and $C'$, the central server can correctly compute the Linear, Gaussian, the Polynomial and the Rational Quadratic Kernels for the distributed input data.
    \end{theorem}
    \begin{proof}
        Without loss of generality, let us consider the kernels to be computed between Alice and Bob. Let us denote an image from Alice as $a = A_i$, and Bob as $b=B_j$. Then representing the encoded images as $a'=A'_i$ and $b'=B'_j$, we propose that the functions to be computed by the central server are as follows.
\begin{gather*}
    K_{lin}^{\dag}(a',b') = [(A'_i)(B'_j)^{\dag}], \quad 
    K_{RBF}^{\dag}(a',b') = \exp \big( -\gamma d_{ij}(A',B') \big),\\
    K_{poly}^{\dag}(a',b') = \left( 1 + [(A'_i)(B'_j)^{\dag}] \right)^d,  \quad 
    K_{RQ}^{\dag}(a',b')=\gamma^2 \left(1+\frac{d_{ij}(A',B')}{2\alpha l^2}\right)^{-\alpha},
\end{gather*}
where
\begin{equation*}
        d_{ij}(A',B'):=\big( [(A'_i)(A'_i)^{\dag}] + [(B'_j)(B'_j)^{\dag}] - [(A'_i)(B'_j)^{\dag}] - [(B'_j)(A'_i)^{\dag}] \big).
    \end{equation*}   
To show that these are well defined kernels, we prove the correctness of the computed functions. Without loss of generality for $P,Q \in \{A,B\}$ and any $p,q$ -
\begin{equation*}
    [(P'_p)(Q'_q)^{\dag}] = \left[(P')\left( {_{\mathcal{Q}}} \Gamma \sigma (Q_q)^{\dag}\right)\right]_p\\
    %&=\left[(P') \left({_Q} O_1 \oplus \dots \oplus  {_Q} O_q\right)(\sigma Q_q)\right]_p\\
    = \left[P \sigma^{\dag}\left(\bigoplus_{k=1}^q {_\mathcal{P}}O_k^{\dag} {_Q}O_k\right)\sigma Q_q\right]_p
    %=\left[P \sigma^{\dag}\sigma Q_q\right]_p
    =\left[P_pQ_q\right].
\end{equation*}
Hence, it follows that 
\begin{gather*}
    [(A'_i)(B'_j)^{\dag}]=ab^T,\quad
    [(B'_j)(A'_i)^{\dag}]=ba^T,\\
    [(A'_i)(A'_i)^{\dag}]=aa^T,\quad
    [(B'_j)(B'_j)^{\dag}]=bb^T,\\
    K^{\dag}_{lin}(a',b')=K_{lin}(a,b),\quad
    K^{\dag}_{RBF}(a',b')=K_{RBF}(a,b),\\
    K^{\dag}_{poly}(a',b')=K_{poly}(a,b),\quad
    K^{\dag}_{RQ}(a',b')=K_{RQ}(a,b).
\end{gather*}

This shows the correctness of the computed functions, hence theoretically satisfying our \textbf{\emph{Correctness}} requirement. Further, to show that these functions are well defined kernels, we note that the map $\Gamma:\mathbb{R}^f\to \mathbb{C}^{f + k}$ defined by $_{\mathcal{P}}\Gamma(v)={v}'$ for $v\in\mathbb{R}^f$ and $P \in \{\mathcal{A},\mathcal{B},\mathcal{C}\}$ is a bijection, and since $K_{lin}$, $K_{RBF},$ $K_{poly}$ and $K_{RQ}$ are well defined kernels, so are $K_{lin}^{\dag}=K_{lin}\circ(\Gamma^{-1}\times\Gamma^{-1}),$ $K_{RBF}^{\dag}=K_{RBF}\circ(\Gamma^{-1}\times\Gamma^{-1}),$ $K_{poly}^{\dag}=K_{poly}\circ(\Gamma^{-1}\times\Gamma^{-1}),$ and $K_{RQ}^{\dag}=K_{RQ}\circ(\Gamma^{-1}\times\Gamma^{-1})$ on the range of $\Gamma\times\Gamma$.
% https://q.uiver.app/#q=WzAsNCxbMCwxLCJcXG1hdGhiYntSfV5mXFx0aW1lc1xcbWF0aGJie1J9XmYiXSxbMCwyLCJcXG1hdGhjYWx7SH1cXHRpbWVzXFxtYXRoY2Fse0h9Il0sWzAsMCwiXFxtYXRoYmJ7Un1ee2Yra31cXHRpbWVzIFxcbWF0aGJie1J9XntmK2t9Il0sWzEsMSwiXFxtYXRoYmJ7Un0iXSxbMCwyLCJcXEdhbW1hXFx0aW1lc1xcR2FtbWEiXSxbMCwxLCJcXHBoaVxcdGltZXMgXFxwaGkiLDJdLFsyLDEsIihcXEdhbW1hXnstMX1cXGNpcmNcXHBoaSlcXHRpbWVzKFxcR2FtbWFeey0xfVxcY2lyY1xccGhpKSIsMix7Im9mZnNldCI6NSwiY3VydmUiOjUsInNob3J0ZW4iOnsidGFyZ2V0IjoxMH19XSxbMSwzLCJcXGxhbmdsZVxcY2RvdCxcXGNkb3RcXHJhbmdsZSIsMl0sWzAsMywiSyIsMV0sWzIsMywiS157XFxkYWd9IiwxXV0=
\[\begin{tikzcd}
	{\mathbb{C}^{f+k}\times \mathbb{C}^{f+k}} \\
	{\mathbb{R}^f\times\mathbb{R}^f} & {\mathbb{R}} \\
	{\mathcal{H}\times\mathcal{H}}
	\arrow["\Gamma\times\Gamma", from=2-1, to=1-1]
	\arrow["{\phi\times \phi}"', from=2-1, to=3-1]
	\arrow["{(\Gamma^{-1}\circ\phi)\times(\Gamma^{-1}\circ\phi)}"', shift right=5, curve={height=24pt}, shorten >=4pt, from=1-1, to=3-1]
	\arrow["{\langle\cdot,\cdot\rangle}"', from=3-1, to=2-2]
	\arrow["K"{description}, from=2-1, to=2-2]
	\arrow["{K^{\dag}}"{description}, from=1-1, to=2-2]
\end{tikzcd}\]
 Hence we have the above commutative diagram. 
\end{proof}

\subsection{Data Adaptability}
In our methodology, Alice, Bob, and Charlie contribute datasets $ A $, $ B $, and $ C $. Introducing additional data $ X $ from Charlie or another participant is seamless. The shared seed is used to generate $ X' $ for further kernel computations with the existing data. Only specific kernel portions are updated, saving computational resources. When data changes dynamically, OKRA creates new masks, denoted by $ \Gamma_t $ for training iteration $ t $, ensuring privacy for each party's data alterations. If a participant withdraws data, the corresponding contribution is removed, complying with data protection standards such as the GDPR.

%% file: 4-privacy.tex
\section{Privacy Analysis}
\label{sec:privacy}
As stated earlier, OKRA operates under the environment consisting of a proper subset of semi-honest participants, and/or a non-colluding semi-honest central server. For this purpose, we will show that it guarantees security in both settings. 

It is essential to highlight that our privacy model excludes certain extreme adversarial conditions. We exclude data distributions where the number of features or the training data of one or more input parties can be guessed, and protocols with only one input party. 

\begin{theorem} \label{theorem3}
    The proposed methodology is secure against a semi-honest adversary who corrupts the central server.
    \end{theorem}
\vspace{-0.2cm}
\begin{proof}
A semi-honest central server is only the receiver of the encoded data from the input parties, and follows the protocol as intended. Without loss of generality, let there be two input parties Alice and Bob with input data $A \in \mathbb{R}^{n_{\mathcal{A}} \times f}$ and $B \in \mathbb{R}^{n_{\mathcal{B}} \times f}$, respectively, where $n_\mathcal{P}$ is the number of samples with the corresponding party $\mathcal{P}\in\{\mathcal{A},\mathcal{B}\}$ and $f$ is the number of features in each image. The semi-honest function party receives the projected data from Alice and Bob. These consist of $A'=A\gamma^\dag \in \mathbb{C}^{n_{\mathcal{A}} \times (f+k)}$ and $B'=B\gamma^\dag \in \mathbb{C}^{n_{\mathcal{B}} \times (f+k)}$ where $k > 0$. Then, it computes the kernel functions mentioned above in Theorem \ref{correctness}. The data to which the function party has access then includes 
\begin{compactenum}
    \item[$(a)$] $A'$ and analogously, $B'$,
    \item[$(b)$] $AB^T=(BA^T)^T$, $AA^T$ and analogously $BB^T$.
\end{compactenum}
Regarding $(a)$, it is evident that $A'$ does not reveal the number of features of $A$. We now show that $A'$ is not produced by a unique pair $A$ and $_{\mathcal{A}}\Gamma\circ \sigma$.

Given a unitary matrix $U \in \mathbb{C}^{f\times f}$ with $f>1$, for $\tilde{A}=AU$ and $_{\mathcal{A}} \tilde{\Gamma}=\Gamma \circ \sigma \circ U$, we have $A'=A(_{\mathcal{A}}\Gamma\circ\sigma)^\dag = \tilde{A}_{\mathcal{A}}\tilde{\Gamma}^\dag$. In the context of reconstruction attacks using analogous data, consider a sample dataset $\hat{A}$ that is drawn from a distribution similar to $A$. In the context of reconstructing $A$, the goal is to find a transformation matrix $R$ such that $A'R=\hat{A}$ approximates $A$. However, since $\hat{A}$ and $A$ do not exist within the same feature space, this leads to loss of information, rendering the approximation ineffective. Further, since $\Gamma$ consists of orthogonal vectors, the inability to approximate $A$ can also be attributed to the unbalanced orthogonal procrustes problem which is unsolved \cite{gower2004procrustes}.

Regarding $(b)$, we adopt the proof from \cite{unal2021escaped}. The matrices that produce these symmetric matrices are not unique, since for any unitary matrix $U \in \mathbb{C}^{(f+k)\times (f+k)}$ where $f>1$, labeling $\tilde{A}=A'U$ and $\tilde{B}=B'U$, we have
\begin{gather*}
    \tilde{A}\tilde{A}^T=AA^T,\quad
        \tilde{B}\tilde{B}^T=BB^T,\quad
                \tilde{A}\tilde{B}^T=AB^T.
\end{gather*}
Similar to the prior argument, these results can also be derived from data in a different feature space. As a result, the function party is limited to learning only the singular values and vectors of the matrices. This means it can derive $U \text{ and } S$
from the singular value decomposition $ A = USV^T $ by eigen-decomposing $ AA^T $. Yet, this information is inadequate to deduce $A$ since $V$ is unknown. Thus, the function party does not learn $\text{(i) the input data or (ii) the number of} $ $\text{the features.} $
\end{proof} 
\begin{theorem}
    The proposed methodology is secure against a proper subset of semi-honest participants. 
\end{theorem}
\begin{proof}
    Since the proposed methodology follows one-shot federated learning with no communication back from the central server, the uni-directional flow of data prohibits the corrupt participants from learning anything about the data of the non-corrupt participants. 
\end{proof}

Therefore, we've shown that the data of non-corrupt participants can not be learned by corrupt semi-honest participants, and that the central server doesn't learn the image sizes, satisfying our \textbf{\emph{Privacy}} requirement.

%% file: 5-experiments.tex
\section{Experiments} \label{sec:experiments}

Experiments were conducted on an AMD 7713 at 2.0GHz with one CPU core and 16 GB of RAM, simulating hospitals collaboratively performing kernel-learning with limited resources. Communication settings included a bandwidth of $1.25$MBps, an average latency of $0.1s$, and a packet loss of $2\%$. 
Unless stated otherwise, OKRA was tested with three participants. Each participant operated as an independent process using TCP connections. Data encoding and transmission to the central server were done via multiple threads. The server performed kernel functions on aggregated data for machine learning model training.
Experiments were deemed successful if OKRA's accuracy matched non-private classification methods and surpassed the run-times of state-of-the-art randomized encoding algorithms, ESCAPED ~\cite{unal2021escaped} and FLAKE ~\cite{hannemann2023privacy}.

\subsection{Datasets}

To demonstrate OKRA's applicability, we experimented with two clinical image datasets with different label types and distributions. The first dataset consists of preprocessed MRI images capturing various stages of Alzheimer's disease \cite{lamontagne2019oasis}. The second dataset features augmented images of white blood cells, representing four distinct cell types \cite{BCCD_dataset}. 
For run-time analysis, we used synthetic datasets, essential for benchmarking without real-world data irregularities. The generated data is balanced, with multi-class clusters drawn from a multivariate normal distribution (mean $0$, variance $1$ for each attribute). We generated $400$ samples per participant with $1000$ features.

\subsection{Correctness of the Model}

Before beginning the run-time experiments, we aimed to validate OKRA's correctness on kernel learning methods. We conducted classification using SVM and dimensionality reduction with Kernel-PCA. Each experiment was run 10 times, testing various combinations of datasets, methods, and kernels (Gaussian, Polynomial, and Linear). Seeds were used for deterministic behavior and reproducible results.

For the classification task, OKRA-SVM was tested against a Naive SVM. Both implementations shared identical hyperparameters $C \in \{2^{-4}, ..., 2^{10}\}$ (misclassification penalty) and $p \in \{1, ..., 5\}$ (degree), optimized using Grid Search. Both were trained with a 5-fold cross-validation. For this implementation, we used Sequential Minimal Optimization (libsvm) provided by scikit-learn \cite{zeng2008fast}. We applied Macro Averaging when the datasets had an unbalanced distribution of classes. Correspondingly, we have evaluated the balanced datasets with Micro Averaging. Table~\ref{table:scores_svm} shows that OKRA-SVM and the naive classifier produce the same F1 and ROC AUC scores.

\begin{wraptable}{r}{0.4\textwidth} % Adjust the width as needed
\centering
% \vskip
\vskip -0.4in
\caption{Performance metrics for OKRA-PCA}
\vskip 0.1in
\label{table:scores_pca}
\footnotesize % Smaller font size
\begin{tabular}{>{\centering\arraybackslash}m{1.8cm}@{\hskip 3pt} >{\centering\arraybackslash}m{1.6cm}@{\hskip 3pt} c@{\hskip 5pt} c}
\toprule
\multirow{2}{*}{Dataset} & \multirow{2}{*}{Kernel} & \multicolumn{2}{c}{OKRA} \\
\cmidrule(lr){3-4}
& & MSE & $R^2$ \\
\midrule
Alzheimer's Disease & Polynomial & 0.0 & 1.0 \\
\addlinespace[0.5ex]
Blood Cell & Gaussian & 0.0 & 1.0 \\
\bottomrule
\end{tabular}
%\vskip -0.1in

\end{wraptable}

For the dimensionality reduction, we compared a naive Kernel-PCA and OKRA-PCA. A kernel is computed, followed by eigenvalue decomposition to sort eigenvalues and eigenvectors by magnitude. The top $n$ eigenvectors are then chosen to transform the data into the new principal component space. To assess utility, the principal components of both Kernel-PCA and OKRA-PCA, utilizing the same kernel, were compared by computing the Mean Squared Error (MSE) and R-squared (R2). 

\begin{table*}[ht!]
\centering
\small % Reduce font size
\resizebox{\textwidth}{!}{%
\begin{tabular}{>{\centering\arraybackslash}m{3.5cm}@{\hskip 5pt} >{\centering\arraybackslash}m{2cm}@{\hskip 3pt} >{\centering\arraybackslash}m{2cm}@{\hskip 5pt} c@{\hskip 5pt} c@{\hskip 5pt} c@{\hskip 5pt} c}
\toprule
\multirow{2}{*}{\begin{tabular}[c]{@{}c@{}}Dataset\\ (Samples $\times$ Features)\end{tabular}} & \multirow{2}{*}{Label Type} & \multirow{2}{*}{Kernel} & \multicolumn{2}{c}{Naive} & \multicolumn{2}{c}{OKRA} \\
\cmidrule(lr){4-5} \cmidrule(lr){6-7}
& & & F1 & ROC AUC & F1 & ROC AUC \\
\midrule
\multirow{2}{*}{\begin{tabular}[c]{@{}c@{}}Alzheimer's Disease \\ \small{$(6400 \times 256 \cdot 256)$}\end{tabular}} & Binary & Gaussian  & $0.91\pm 0.03$ & $0.97\pm 0.01$ & $0.91\pm 0.03$ & $0.97\pm 0.01$ \\
& Multi-class & Gaussian & $0.97\pm 0.04$ & $1.00\pm 0.00$ & $0.97\pm 0.04$ & $1.00\pm 0.00$ \\
\addlinespace[0.5ex]
Blood Cell \\ \small{$(12500 \times 100 \cdot 100 \cdot 3)$} & Multi-class & Gaussian & $0.88\pm 0.01$ & $0.94\pm 0.01$ & $0.88\pm 0.01$ & $0.94\pm 0.01$ \\
\addlinespace[0.5ex]
Synthetic \\ \small{$(200-51200 \times 12800)$} & Multi-class & Linear & $0.89\pm 0.03$ & $0.95\pm 0.02$ & $0.89\pm 0.03$ & $0.95\pm 0.02$ \\
\bottomrule
\end{tabular}%
}
\vskip 0.1in
\label{table:scores_svm}
\vskip -0.2in
\end{table*}

This provides a quantitative comparison of the two PCA methods in terms of their ability to capture the data's variance and structure.
Table~\ref{table:scores_pca} indicates that both OKRA-PCA and naive PCA are highly consistent and produce identical results, explaining the same patterns in the data. 
As expected, the encoding has no influence on the results.  Hence we experimentally satisfy the \textbf{\emph{Correctness}} requirement.

\subsection{Performance Analysis}
To ensure OKRA's computational overhead is manageable, especially with high-dimensional datasets like images, we compared its run-time with FLAKE and ESCAPED. We focused on two run-time types: \emph{Encoding time}, the overhead of encoding data using OKRA for an individual participant, and \emph{Total time}, which includes communication from when the server initiates connections to the completion of the kernel function computation. Our analysis examined the impact of varying the number of participants and image sizes.
i
\paragraph{Scaling up the Number of Participants}
Evaluating our method's adaptability across multiple participants is crucial. Figure \ref{scaling_parties} shows the cumulative time required for data encoding, transmission, and linear kernel function computation. Each participant encodes a dataset of 400 samples with 1000 features. OKRA performs exceptionally well: with 10 participants, the execution takes only 16 seconds on average.

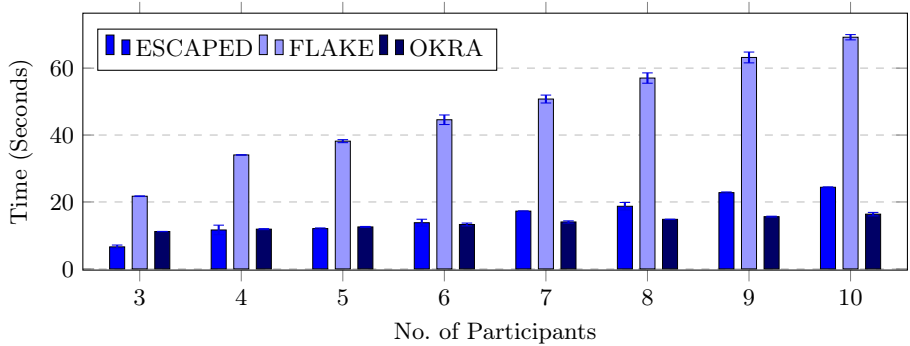
\begin{figure*}[ht!] 
    \centering
        \vskip 0.2in
    \begin{tikzpicture}
        \begin{axis}[
            width=\linewidth,
            height=5cm,
            ybar, % Bar chart
            enlarge x limits=0.08,
            legend style={at={(0.26,0.94)},
            anchor=north,legend columns=-1},
            ylabel={Time (Seconds)},
            xlabel={No. of Participants},
            symbolic x coords={3, 4, 5, 6, 7, 8, 9, 10},
            xtick=data,
            bar width=0.2cm,
            x=1.35cm,
            error bars/.cd,
            y dir=normal,
            ymajorgrids=true,
            grid style=dashed
            ]
            \addplot+[
                    bar shift=-0.3cm,
                    draw = black,
	                line width = .1mm,
	                fill = pgfplotsmidblue,
                    error bars/.cd,
                    y dir=both,
                    y explicit,
                    error bar style={blue!90!black, thick}] 
            coordinates {
                (3, 6.577577) +- (0.357606535,0.546422615)
                (4, 11.6310785) +- (3.48850061,1.435517054)
                (5, 12.039910957951155) +- (0.06184832825144743, 0.1998005651048747)
                (6, 13.8027093410491) +- ( 0.4266516952, 1.033912897)
                (7, 17.26282247) +- (0.032124782, 0.059115648)
                (8, 18.71406213) +- (0.359431466, 1.136866808)
                (9, 22.78567847) +- (0.125785325, 0.20749712)
                (10, 24.39194409) +- (0.176846808, 0.113780975)
            };
            \addlegendentry{ESCAPED}

            \addplot+[
                    draw = black,
	                line width = .1mm,
	                fill = pgfplotslightblue,
                    error bars/.cd,
                     y dir=both,
                     y explicit,
                     error bar style={blue!90!black, thick}] 
            coordinates {
                (3, 21.71630211) +- ( 0.151531245, 0.087017926)
                (4, 34.05785904) +- (0.503899775, 0.023131999)
                (5, 38.17623015) +- (0.145516358, 0.449171255)
                (6, 44.57180891) +- ( 0.62123542242, 1.4271603522)
                (7, 50.76079378) +- ( 0.37415312542, 1.1823627942)
                (8, 57.03156147) +- ( 0.4115885552, 1.5449171252)
                (9, 63.18762646) +- ( 0.9887654214, 1.61455163542)
                (10, 69.2336031) +- ( 0.8659455562, 0.78845698542)
            };
            \addlegendentry{FLAKE}            

            \addplot+[
                    bar shift=0.3cm,
                    draw = black,
	                line width = .1mm,
	                fill = pgfplotsdarkblue,
                    error bars/.cd,
                     y dir=both,
                     y explicit,
                     error bar style={blue!90!black, thick}]  
            coordinates {
                (3, 11.16575024) +- (0.05004086263427729,0.031320669577025746)
                (4, 11.8755928) +- (1.1015799747988,0.14512748778664902)
                (5, 12.54128449) +- (0.06426182037262977,0.1375182604226578)
                (6, 13.27032421) +- ( 0.42, 0.42)
                (7, 14.02492967) +- ( 0.125225462, 0.325648)
                (8, 14.80615208) +- ( 0.4478523666, 0.036589632)
                (9, 15.6067704) +- ( 0.025687474, 0.135695566)
                (10, 16.36286068) +- ( 0.316175318, 0.513215685)
            };
            \addlegendentry{OKRA}
        \end{axis}
    \end{tikzpicture}
    \vskip -0.1in
    \caption{Runtime comparison with increasing participants.}
    \label{scaling_parties}
    \vskip -0.1in
\end{figure*}

\paragraph{Scaling up the Image Size}

Medical image data, essential for accurate diagnoses and treatment, is typically captured at high resolutions. However, for preprocessing, these images are often scaled down for memory and computational efficiency, with rescaled sizes ranging from $32 \times 32$ up to $256 \times 256$. To assess whether OKRA can encode preprocessed medical images without burdening the data holder, we experimented with increasing image sizes for one participant with 400 images (refer to Fig. \ref{fig:scaling_images}). The results demonstrate that OKRA is applicable to medical imaging, whereas state-of-the-art methods are not. We, therefore, meet both requirements \textbf{\emph{Data Adaptability}} and \textbf{\emph{Efficiency}}.

%For ML tasks, this results in $192,768,3072,12288,49152,196608$ features. 

%The overhead of randomized encoding consists of two run-times: the encoding time and the kernel function computation time. 

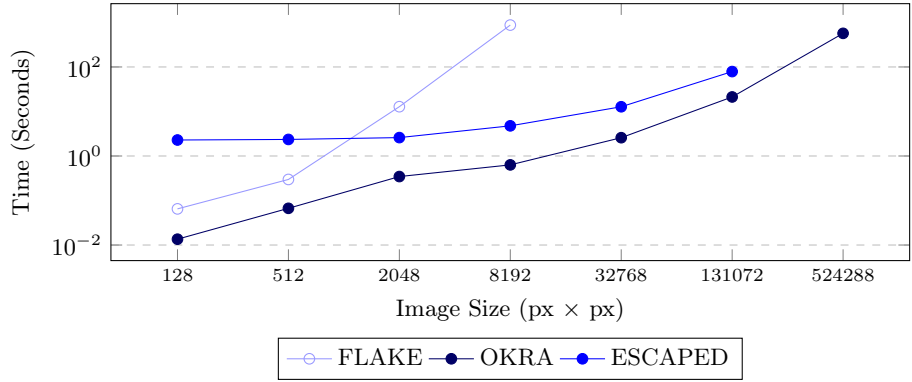
\begin{figure}[h!]
    \centering
    \begin{tikzpicture}
        \begin{axis}[
            width=\linewidth,
            height=5cm,
            xmode=log,
            ymode=log,
            xlabel={Image Size (px $\times$ px)},
            ylabel={Time (Seconds)},
            xtick={128, 512, 2048, 8192, 32768, 131072, 524288},
            xticklabels={\scriptsize 128, \scriptsize 512, \scriptsize 2048, \scriptsize 8192, \scriptsize 32768, \scriptsize 131072, \scriptsize 524288}, % Adjust the font size here
            legend style={at={(0.5,-0.3)}, anchor=north,legend columns=-1},
            ymajorgrids=true,
            grid style=dashed,
            ]
            
            \addplot[mark=o,pgfplotslightblue] coordinates {
                (128, 0.06489300727844238)
                (512, 0.2971668243408203)
                (2048, 12.81406044960022)
                (8192, 872.997763633728)
                %(32768, 42)
                %(131072, 42)
                %(524288, 42)
            };
            
            \addlegendentry{FLAKE}
            
            \addplot[mark=*,pgfplotsdarkblue] coordinates {
                (128, 0.013526201248168945)
                (512, 0.06677818298339844)
                (2048, 0.34534215927124023)
                (8192, 0.6316003799438477)
                (32768, 2.5774924755096436)
                (131072, 21.19995355606079)
                (524288, 569.0589525699615)
            };
            \addlegendentry{OKRA}

            \addplot[mark=*,pgfplotsmidblue] coordinates {
                (128, 2.2869815826416016)
                (512, 2.35477614402771)
                (2048, 2.590272903442383)
                (8192, 4.758647680282593)
                (32768, 12.765627384185791)
                (131072, 78.79857230186462)
                %(524288, 42)
            };
            \addlegendentry{ESCAPED}
        \end{axis}
    \end{tikzpicture}
    \caption{Encoding times against image sizes.}
    \label{fig:scaling_images}
    \vskip -0.2in
\end{figure}

\subsection{Discussion}
Randomized encoding methods, used for privacy-preserving kernel learning, offer the advantage of computational efficiency while maintaining model accuracy. In the field of randomized encoding methods, our proposed solution, OKRA, demonstrates significant advantages. It consistently outperforms both FLAKE and ESCAPED without compromising utility. Notably, OKRA minimizes the need for extended communication rounds and supports data adaptability, issues primarily encountered with ESCAPED. Additionally, OKRA achieves more efficient data encoding than FLAKE and can handle larger image sizes. With these strengths, OKRA emerges as a robust and efficient option within the domain of randomized encoding techniques for image data.

%ESCAPED offers a precise solution, but at the cost of increased communication between participants. This interaction often leads to extended execution times. Moreover, a notable limitation is also its inability to adapt seamlessly to new parties or data, hindering its flexibility in dynamic environments.

%FLAKE, on the other hand, can compute an exact kernel matrix, which further allows the computation of any desired kernel matrix. This process facilitates training any kernel learning algorithm as if the data were centralized. However, the masking process is computationally demanding for large amounts of data or larger image sizes. Here, FLAKE experiences substantial computational overhead, making it less efficient in scenarios demanding higher privacy against reconstruction attacks. 

In subsequent work, we aim to explore the challenges associated with the colluding central server scenario. Further, we intend to address malicious participant behavior, where encoded data is altered to cause inaccuracies in the server's results. Additionally, we plan to explore the applicability of OKRA to other kernel methods, given its promising potential. 

%% file: 6-conclusion.tex
\section{Conclusion} \label{sec:conclusion}

%need to include medical imaging into this
Privacy preservation for the analysis of distributed medical images remains a critical challenge. While randomized encoding offers a promising avenue, existing methodologies often grapple with the trade-off between accuracy and efficiency. Addressing this gap with regards to kernel functions applied to high-dimensional data, we introduced OKRA, a novel algorithm operating within a one-shot federated learning paradigm. OKRA not only privately computes exact kernel functions but also ensures that kernel-based machine learning algorithms are trained efficiently. Empirical evaluations further validate the superiority of OKRA, demonstrating its potential to achieve accuracy levels comparable to centralized machine learning models while minimizing computational burdens. As the landscape of machine learning continues to evolve, OKRA's robustness and efficiency position it as a pivotal tool in privacy-preserving endeavors.